\documentclass[a4paper,11pt]{article}

\usepackage[utf8]{inputenc}
\usepackage[margin=3cm]{geometry}
\usepackage{graphicx}
\usepackage{subcaption}
\usepackage{amsmath, amssymb, amsthm}
\usepackage{booktabs}
\usepackage{enumerate,paralist}
\usepackage[dvipsnames]{xcolor}

\usepackage[nocompress,noadjust]{cite}
\usepackage[pdfpagelabels,colorlinks,allcolors=MidnightBlue]{hyperref}
\usepackage[nameinlink,capitalize,noabbrev]{cleveref}

\newtheorem{theorem}{Theorem}
\newtheorem{lemma}{Lemma}
\newtheorem{cor}{Corollary}

\newtheorem{open}{Open Problem}

\graphicspath{{pics/}}

\newcommand{\df}[1]{{\it #1}}
\newcommand{\Oh}{{\ensuremath{\mathcal{O}}}}
\newcommand{\Sh}{{\ensuremath{\mathcal{S}}}}
\newcommand{\Qh}{{\ensuremath{\mathcal{Q}}}}

\DeclareMathOperator{\sn}{sn}
\DeclareMathOperator{\dsn}{dsn}
\DeclareMathOperator{\qn}{qn}
\DeclareMathOperator{\sqn}{sqn}

\DeclareMathOperator{\pw}{pw}
\DeclareMathOperator{\tw}{tw}

\newcommand{\cbox}{%
	\text{\;\fboxsep=-.2pt\fbox{\rule{0pt}{1.5ex}\rule{1.5ex}{0pt}}\;}%
}

\newcounter{dummycount}

\newcommand{\wormholeThm}[1]{
	\newcounter{#1}
	\setcounter{#1}{\value{theorem}}}

\newenvironment{backInTimeThm}[1]{
	\setcounter{dummycount}{\value{theorem}}
	\setcounter{theorem}{\value{#1}}}
{\setcounter{theorem}{\value{dummycount}}}

\begin{document}

\title{Book Embeddings of Graph Products}

\author{
	Sergey Pupyrev \\
	spupyrev@gmail.com
}

\date{}
\maketitle

\begin{abstract}
A $k$-stack layout (also called a $k$-page book embedding) of a graph consists of a total order of the vertices, 
and a partition of the edges into $k$ sets of non-crossing edges with respect to the vertex order.
The stack number (book thickness, page number) of a
graph is the minimum $k$ such that it admits a $k$-stack layout.
A $k$-queue layout is defined similarly, except that no two edges in a single set may be nested.

It was recently proved that graphs of various non-minor-closed classes
are subgraphs of the strong product of a path and a graph with bounded treewidth. Motivated by
this decomposition result, we explore stack layouts of graph products. We show that the
stack number is bounded for the strong product of a path and 
(i) a graph of bounded pathwidth or (ii) a bipartite graph of bounded treewidth and bounded degree.
The results are obtained via a novel concept of simultaneous stack-queue layouts, which
may be of independent interest.
\end{abstract}

\section{Introduction}
\label{sect:intro}

Embedding graphs in books is a fundamental problem in graph theory, which has
been the subject of intense research since their introduction in 70s by Ollmann~\cite{Oll73}.
A \df{book embedding} (also known as a \df{stack layout}) of a graph $G=(V, E)$
consists of a total order, $\sigma$, of $V$ and an assignment of the
edges to \df{stacks} (\df{pages}), such that no two edges in a single stack \df{cross};
that is, there are no edges $(u, v)$ and $(x, y)$ in a stack with $u <_{\sigma} x <_{\sigma} v <_{\sigma} y$.
The minimum number of pages needed for a book embedding of a graph $G$
is called its \df{stack~number} (or \df{book thickness} or \df{page~number}) and denoted by $\sn(G)$.

Book embeddings have been extensively studied for various families of graphs. 
In particular, the graphs with stack number one are precisely the outerplanar graphs, while the graphs with 
stack number at most two are the subgraphs of planar Hamiltonian graphs~\cite{BK79}.
The stack number of planar graphs is four~\cite{four}, 
graphs of genus $g$ have stack number $\Oh(\sqrt g)$~\cite{Mal94}, while for graphs of
treewidth $tw$, it is at most $tw+1$~\cite{GH01}.
More generally, all proper minor-closed graph families have a bounded stack number~\cite{Bla03}.
Non-minor-closed classes of graphs have also been investigated.
Bekos et al. proved that $1$-planar graphs have bounded stack number~\cite{BBKR17}. 
Recall that a graph is $k$-planar if it
can be drawn in the plane with at most $k$ crossings per edge.
Recently the result has been generalized to
a wider family of $k$-framed graphs that admit a drawing with a planar skeleton, 
whose faces have degree at most $k \ge 3$ and whose crossing edges are in the interiors of the faces~\cite{BDGGMR20}.
In general however, the best-known upper bound on the stack number of $k$-planar graphs is $\Oh(\log n)$~\cite{DF18}.

\begin{figure}[!tb]
	\begin{subfigure}[b]{0.3\textwidth}
		\centering
		\includegraphics[page=3,width=0.8\textwidth]{products}
		\caption{$P_4 \cbox G$}
	\end{subfigure}
	\hfill
	\begin{subfigure}[b]{0.3\linewidth}
		\centering
		\includegraphics[page=2,width=0.8\textwidth]{products}
		\caption{$P_4 \times G$}
	\end{subfigure}
	\hfill
	\begin{subfigure}[b]{0.3\linewidth}
		\centering
		\includegraphics[page=1,width=0.8\textwidth]{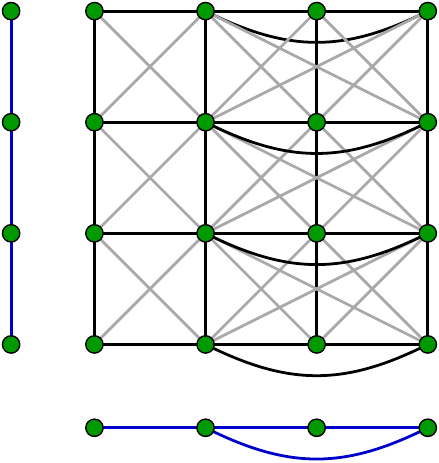}
		\caption{$P_4 \boxtimes G$}
	\end{subfigure}
	\caption{Examples of graph products of a path, $P_4$, and a cycle with an edge,~~$G$: 
		(a)~cartesian, (b)~direct, (c)~strong.}
	\label{fig:products}
\end{figure}

We suggest to attack the problem of determining book thickness of non-planar
graphs using graph products. Formally, let $A$ and $B$ be graphs. A product of $A$ and $B$ is a graph
defined on a vertex set
$$
V(A) \times V(B) = \{(v, x) : v \in V(A), x \in V(B)\}).
$$
A potential edge, $(v, x), (u, y) \in V(A) \times V(B)$, can be classified as follows:
\begin{itemize}
	\item $A$-edge: $v = u$ and $(x, y) \in E(B)$, or
	\item $B$-edge: $x = y$ and $(v, u) \in E(A)$, or
	\item $direct$-edge: $(v, u) \in E(A)$ and $(x, y) \in E(B)$.
\end{itemize}	

\noindent
The \df{cartesian product} of $A$ and $B$, denoted by $A \cbox B$, consists of $A$-edges and $B$-edges.
The \df{direct product} of $A$ and $B$, denoted by $A \times B$, consists of direct edges.
The \df{strong product} of $A$ and $B$, denoted by $A \boxtimes B$, consists of $A$-edges, $B$-edges, 
and direct-edges. 
\cref{fig:products} illustrates examples of the defined
graph products. 
Notice that all the products are symmetric. In this paper, we study stack layouts of strong
products of a path and a bounded-treewidth graph (refer to \cref{sect:prel} for a definition), 
focusing primarily on the following question:

\begin{open}
	\label{open:stack_strong_product}
	Is stack number of $P_n \boxtimes G$, where $P_n$ is a path and $G$ is a graph of treewidth $tw \ge 1$, 
	bounded by $f(tw)$ for some function $f$?
\end{open}

Our motivation for studying stack layouts of graph products comes from a recent development 
of decomposition theorems for planar and beyond-planar graphs~\cite{PS19,DJMMUW19,DMW20}. 
Dujmovi{\'c}, Morin, and Wood~\cite{DMW20} recently show the following:

\begin{lemma}[\cite{DMW20}]
	\label{lm:kplanar_decom}
	Every $k$-planar graph is a subgraph of the strong product of a path and a graph of treewidth $\Oh(k^5)$.
\end{lemma}

Notice that \cref{lm:kplanar_decom} together with an affirmative answer to \cref{open:stack_strong_product}
would provide a constant stack number for all $k$-planar graphs, thus resolving a long-standing 
open problem listed in a recent survey on graph drawing of beyond-planar graphs~\cite{DLM19}.
Furthermore, a similar decomposition exists for other classes of non-minor-closed families of graphs, such as 
map graphs, string graphs, graph powers, and nearest neighbor graphs, whose
stack number is not known to be bounded by a constant; refer to \cite{DMW20} for exact definitions.
Interestingly, a negative answer to \cref{open:stack_strong_product} would resolve another 
question in the context of queue layouts that remains unsolved for more than thirty years.

A \df{queue layout} is a ``dual'' concept of a stack layout. For a graph $G=(V, E)$, it consists 
of a total order, $\sigma$, of $V$ and an assignment of the
edges to \df{queues}, such that no two edges in a single queue \df{nest};
that is, there are no edges $(u, v)$ and $(x, y)$ in a queue with $u <_{\sigma} x <_{\sigma} y <_{\sigma} v$.
The minimum number of queues needed in a queue layout of a graph
is called its \df{queue number} and denoted by~$\qn(G)$~\cite{HR92}.
As with stack layouts, the queue number is known to be bounded for many classes of graphs, 
including planar graphs~\cite{DJMMUW19}, graphs with bounded treewidth~\cite{DMW05,Wie17}, and all proper minor-closed 
classes of graphs~\cite{DJMMUW19,DMW20}.
Queue layouts have been introduced by Heath, Leighton, and Rosenberg~\cite{HR92,HLR92}, who tried to 
measure the power of stacks and queues to represent a given graph. Despite a wealth of research on the topic,
a fundamental question of what is more ``powerful'' remains unanswered. That is, Heath et al.~\cite{HLR92} ask
whether the stack number of a graph is bounded by a function of its queue number, and whether the queue number
of a graph is bounded by a function of its stack number.
In a study of queue layouts of graph products, Wood~\cite{W05} shows that 
for a path $P_n$ and all graphs $G$, $\qn(P_n \boxtimes G) \le 3 \qn(G) + 1$.
This result together with a negative answer to \cref{open:stack_strong_product} would provide an example
of a graph (namely, the strong product of a path and a bounded-treewidth graph) that has a constant queue
number but an unbounded stack number; thus, resolving one direction of the question posed by Heath et al.~\cite{HLR92}.

\subsection*{Results and Organization}
In this paper we introduce and initiate an investigation of \cref{open:stack_strong_product}. 
Our contribution is twofold. Firstly, we resolve the problem in affirmative for two subclasses
of bounded-treewidth graphs. Secondly, we provide an evidence that the most ``natural'' approach
cannot lead to a positive answer of the problem.

\paragraph{Positive Results.}
It is easy to verify that the stack number of $P_n \boxtimes G$ is bounded by a constant when
$G$ is a ``simple'' graph such as a path, a star, or a cycle. Notice that the strong graph product consists of
$n$ \df{copies} of $G$, which are connected by \df{inter-copy} edges. A {\it natural approach} is to layout each
copy independently using a constant number of stacks and then join individual results into
a final stack layout. In order to be able to embed inter-copy edges in a few stacks, one has to 
alternate direct and reverse vertex orders for the copies of $G$; refer to \cref{fig:pathpath_stacks} for the process of
embedding $P_n \boxtimes P_m$ in four stacks.

\begin{figure}[!tb]
	\begin{subfigure}[b]{0.31\textwidth}
		\centering
		\includegraphics[page=1,width=1\textwidth]{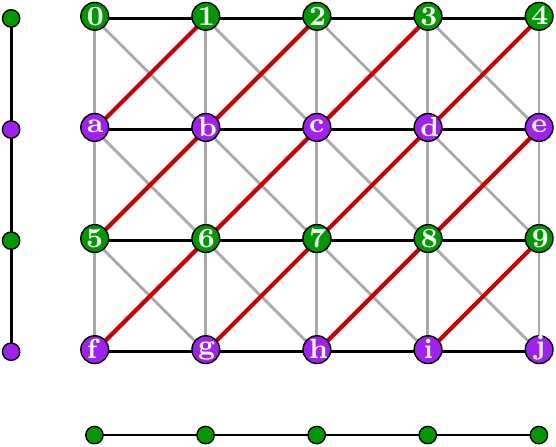}
		\caption{$P_4 \boxtimes P_5$}
	\end{subfigure}
	\hfill
	\begin{subfigure}[b]{0.65\linewidth}
		\centering
		\includegraphics[page=2,width=1\textwidth]{pathpath_stacks}
		\caption{$4$-stack layout}
	\end{subfigure}
	\caption{The strong product of $P_4 \boxtimes P_5$ and its $4$-stack layout. 
		The layout process is easily extendable for $P_n \boxtimes P_m$ with arbitrary values of $n$ and $m$.}
	\label{fig:pathpath_stacks}
\end{figure}

The above technique can be generalized using the concept of simultaneous stack-queue layouts.
Let $\sigma$ be a total order of $V$ for a graph $G=(V, E)$. A \df{simultaneous $s$-stack $q$-queue}
layout consists of $\sigma$ together with (i)~a partition of $E$ into $s$ stacks with respect to $\sigma$, 
and (ii)~a partition of $E$ into $q$ queues with respect to $\sigma$. In such a layout every edge of $G$
is associated with a stack and with a queue.
We stress the difference with so-called \df{mixed} layouts in which an edge belongs to a stack {\it or} to a queue~\cite{P17}.

In order to state the first main result of the paper, we use \df{dispersable}
book embeddings in which the graphs induced by the edges of each page are $1$-regular; see \cref{fig:matching}.
The minimum number of pages needed in a dispersable book embedding of $G$ is called its 
\df{dispersable stack number}, denoted $\dsn(G)$; it is also known as \df{matching book thickness}~\cite{BK79,ABGKP18}.

\wormholeThm{thm-product}
\begin{theorem}
	\label{thm:sn_product}
	Let $H$ be a bipartite graph and $G$ be a graph that admits a simultaneous $s$-stack $q$-queue layout.
	Then
	\begin{compactenum}[(i)]
		\item $\sn(H \cbox G) \le s + \dsn(H)$,
		\item $\sn(H \times G) \le 2 q \cdot \dsn(H)$,
		\item $\sn(H \boxtimes G) \le 2 q \cdot \dsn(H) + s + \dsn(H)$.
	\end{compactenum}
\end{theorem}	

What graphs admit simultaneous layouts for constant $s$ and $q$? We prove that 
graphs of bounded \df{pathwidth} (see \cref{sect:prel} for a definition) have such a layout.
Although it is known that both the stack number and the queue number of pathwidth-$p$ graphs is at most 
$p$~\cite{TY02,DMW05}, \cref{lm:pw_sim} (in \cref{sect:main})
shows that the bounds can be achieved using a common vertex order.
As a direct corollary of the lemma, \cref{thm:sn_product}, and an 
observation that $\dsn(P_n) = 2$, we get the following result.\footnote{Very recently, 
Dujmovi{\'c}, Morin, and Yelle~\cite{DMY20} independently proved 
a result asymptotically equivalent to \cref{cor:1}; see \cref{sect:related} for a discussion.}

\begin{cor}
	\label{cor:1}
	Let $G$ be a graph of pathwidth $p$. Then $\sn(P_n \boxtimes G) \le 5p + 2$.
\end{cor}	

Notice that \cref{cor:1} combined with \cref{lm:kplanar_decom} implies an alternative proof of 
the $\Oh(\log n)$ upper bound for the stack number of
$k$-planar graphs, since for every graph $G$, $\pw(G) \in \Oh(\tw(G) \cdot \log n)$~\cite{Bod98}.

Another corollary of \cref{thm:sn_product} affirmatively resolves \cref{open:stack_strong_product} 
for the strong product of a path and a bounded-treewidth bipartite graph of bounded maximum vertex degree.
For that case we bound the dispersable stack number of a bipartite graph by a function of its
treewidth and the maximum vertex degree; see \cref{lm:disp} in \cref{sect:main}.

\begin{cor}
	\label{cor:2}
	Let $G$ be a bipartite graph of treewidth $tw$ with maximum vertex degree $\Delta$. 
	Then $\sn(P_n \boxtimes G) \le 3 (tw + 1) \Delta + 1$.
\end{cor}	

\paragraph{Negative Results.}

Next we investigate simultaneous stack-queue layouts. We prove that 
if a graph admits a simultaneous $s$-stack $q$-queue layout, then its pathwidth is bounded by
a function of $s$ and $q$. In other words, the class of $\Oh(1)$-pathwidth graphs coincides
with the class of graphs admitting a simultaneous $\Oh(1)$-stack $\Oh(1)$-queue layout.

\wormholeThm{thm-pw}
\begin{theorem}
	\label{thm:sq_pw}
	Let $G$ be a graph admitting a simultaneous $s$-stack $q$-queue layout. 
	Then $G$ has pathwidth at most $2 \cdot s \cdot q$.
\end{theorem}	

Corollaries~\ref{cor:1} and \ref{cor:2} provide sufficient conditions for a graph $G$
to imply a bounded stack number of $P_n \boxtimes G$. Yet many relatively simple graphs of bounded treewidth
(such as trees) have pathwidth $\Omega(\log n)$ and an unbounded vertex degree. 
A reasonable question is whether the conditions are necessary.
Next we study the aforementioned {\it natural approach} to construct stack layouts of graph
products, and prove that it cannot lead to a constant number of stacks for
graphs with an unbounded pathwidth. 
Formally, call a stack layout of $P_n \boxtimes G$ \df{separated}
if for at least two consecutive copies of $G$, $G_1$ and $G_2$, all 
vertices of $G_1$ precede all vertices of $G_2$ in the vertex order.
The next result shows that a separated layout of $P_n \boxtimes G$ with a constant
number of stacks implies a bounded pathwidth of $G$.

\wormholeThm{thm-separated}
\begin{theorem}
	\label{thm:separated}
	Assume $P_n \boxtimes G$ has a separated layout on $s$ stacks. Then $G$ admits 
	a simultaneous $s$-stack $s^2$-queue layout, and therefore, $\pw(G) \le 2 s^3$.
\end{theorem}	

The remaining of the paper is organized as follows. After recalling basic definitions in \cref{sect:prel}, 
we prove the main results of the paper in \cref{sect:main}. 
\cref{sect:related} is devoted to a discussion of related works on 
stack and queue layouts of graph products.
\cref{sect:concl} concludes the paper with possible future directions and interesting open problems.

\section{Preliminaries}
\label{sect:prel}

Throughout the paper, $G = \big(V(G), E(G)\big)$ is a simple undirected graph.
We denote a path with $n$ vertices by $P_n$.
A \df{vertex order}, $\sigma$, of a graph $G$ is a
total order of its vertex set $V(G)$, such that for any two vertices $u$ and $v$,
$u <_{\sigma} v$ if and only if $u$ precedes $v$ in $\sigma$. 
Let $F$ be a set of $k \geq 2$ independent (that is, having no common endpoints) edges 
$(s_i, t_i), 1 \le i \le k$. 
If $s_1 <_{\sigma} \dots <_{\sigma} s_k <_{\sigma} t_k <_{\sigma} \dots <_{\sigma} t_1$, then
$F$ is a \df{$k$-rainbow}, while if 
$s_1 <_{\sigma} \dots <_{\sigma} s_k <_{\sigma} t_1 <_{\sigma} \dots <_{\sigma} t_k$, then
$F$ is a \df{$k$-twist}. Two independent edges forming a $2$-twist ($2$-rainbow) are 
called \df{crossing} (\df{nested}).

A \df{$k$-stack layout} of a graph is a pair $(\sigma, \{\Sh_1, \dots, \Sh_k\})$, 
where $\sigma$ is a vertex order and $\{\Sh_1, \dots, \Sh_k\}$ is a partition of $E(G)$ into 
\df{stacks}, that is, sets of pairwise non-crossing edges.
Similarly, a \df{$k$-queue layout} is $(\sigma, \{\Qh_1, \dots, \Qh_k\})$, 
where $\{\Qh_1, \dots, \Qh_k\}$ is a partition of $E(G)$ into sets of pairwise non-nested edges called
\df{queues}. The minimum number of stacks (queues) in a stack (queue) layout of a graph
is its \df{stack number} (\df{queue number}).
It is easy to see that a $k$-stack layout ($k$-queue layout) cannot have a $k$-twist ($k$-rainbow). 
Furthermore, a vertex order without a $(k+1)$-rainbow corresponds to
a $k$-queue layout~\cite{HR92}. In contrast, a vertex order without a $(k+1)$-twist may not produce
a $k$-stack layout but corresponds to a $f(k)$-stack layout; the best-known function $f$ is
quadratic~\cite{DM19}.

A \df{tree decomposition} of a graph $G$ is given by a tree $T$ whose nodes index a collection
$\big(B_x \subseteq V(G) : x \in V(T)\big)$ of sets of vertices in $G$ called \df{bags} such that:
\begin{compactitem}
	\item For every edge $(u, v)$ of $G$, some bag $B_x$ contains both $u$ and $v$, and
	\item For every vertex $v$ of $G$, the set $\{x \in V(T ) : v \in B_x\}$ induces a non-empty
	connected subtree of $T$.
\end{compactitem}
The \df{width} of a tree-decomposition is $\max_x |B_x|-1$, and the \df{treewidth} of a graph~$G$, 
denoted $\tw(G)$, is the minimum width of any tree decomposition of $G$.

A \df{path decomposition} is a tree decomposition in which the underlying tree, $T$, is a
path. Thus, it can be thought of as a sequence of subsets of vertices, called
\df{bags}, such that each vertex belongs to a contiguous subsequence of bags and each
two adjacent vertices have at least one bag in common.
The \df{pathwidth} of a graph~$G$, 
denoted $\pw(G)$, is the minimum width of any path decomposition of $G$.
We also use an equivalent definition of the pathwidth called the \df{vertex separation number}~\cite{Kin92,Bod98}.
Consider a vertex order $\sigma$ of a graph $G$.
The \df{vertex cut} in $\sigma$ at a vertex $v \in V(G)$ is defined to be 
$C(v) = \{x \in V(G) : \exists (x,y) \in E(G), x <_{\sigma} v \le_{\sigma} y\}$. 
The \df{vertex separation number} of $G$ is the minimum, taken
over all vertex orders $\sigma$ of $G$, of a maximum cardinality of a vertex cut in $\sigma$.

\section{Main Proofs}
\label{sect:main}

\subsection{Positive Results}

\begin{backInTimeThm}{thm-product}
	\begin{theorem}
		Let $H$ be a bipartite graph and $G$ be a graph that admits a simultaneous $s$-stack $q$-queue layout.
		Then
		\begin{compactenum}[(i)]
			\item $\sn(H \cbox G) \le s + \dsn(H)$,
			\item $\sn(H \times G) \le 2 q \cdot \dsn(H)$,
			\item $\sn(H \boxtimes G) \le 2 q \cdot \dsn(H) + s + \dsn(H)$.
		\end{compactenum}
	\end{theorem}	
\end{backInTimeThm}

\begin{figure}[!htb]
	\centering
	\includegraphics[page=3,width=0.8\textwidth]{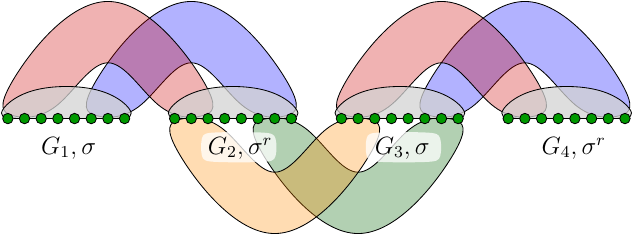}
	\caption{An $(s + 4q + 2)$-stack layout of the strong product of $P_4$ and a graph $G$ that admits a simultaneous 
		$s$-stack $q$-queue layout using vertex order $\sigma$.
		$G_1, G_2, G_3$, and $G_4$ correspond to copies of $G$ laid out by 
		alternating $\sigma$ and its reverse, $\sigma^r$.
		Groups of stacks are colored differently.}
	\label{fig:thm-product}
\end{figure}

\begin{proof}
	For every pair of graphs, the set of edges of the strong product is the union 
	of edges of the cartesian product and the direct product of the graphs. Therefore, claim (iii) of the 
	theorem follows from claims (i) and (ii), which we prove next.
	
	Let $\pi$ be a vertex order of $H$ in the dispersable stack layout, and let $\sigma$ and $\sigma^r$
	be a vertex order of $G$ and its reverse in the simultaneous stack-queue layout.
	We call the parts of the bipartition of $H$ \df{white} and \df{black}, and
	denote by $0 \le \pi(v) < n$ the index of vertex $v \in V(H)$ in $\sigma$.
	To construct an order, $\phi$, for the stack layout of a graph product, 
	we start with $\pi$ and replace each white vertex of $H$ with $\sigma$ and each black vertex of $H$ with $\sigma^r$.
	Formally, for two vertices $(v, x)$ and $(u, y)$ of a product, let $\phi(v, x) < \phi(u, y)$ 
	if and only if 
	\begin{compactenum}[(C1)]
		\item $v \neq u$ and $\pi(v) < \pi(u)$, or 
		\item $v = u$, $v$ is white, and $x <_{\sigma} y$, or
		\item $v = u$, $v$ is black, and $y <_{\sigma} x$.
	\end{compactenum}
	We emphasize that the same vertex order is utilized for all three graph products; 
	see \cref{fig:pathpath_stacks} and \cref{fig:thm-product} for illustrations.

	We first verify that $\sn(H \cbox G) \le s + \dsn(H)$, thus proving claim (i) of the theorem. 
	Since $\sigma$ and $\sigma^r$ are vertex orders
	of an $s$-stack layout of $G$ and different copies of $G$ are separated in $\phi$, 
	all $G$-edges are embedded in $s$ stacks.
	Further, every edge of $H$ is incident to a white and a black vertex of $H$ that correspond to $\sigma$ and $\sigma^r$.
	Thus, $H$-edges between a pair of copies of $G$ are non-crossing and can be assigned to the same stack.
	Since the edges of $H$ require $\dsn(H)$ stacks and each stack consists of independent edges, 
	all $H$-edges are embedded in $\dsn(H)$ stacks.
	
	Next we show that direct-edges can be assigned to $2 q \cdot \dsn(H)$ stacks, which we denote by
	$\Sh_i^j$ for $1 \le i \le q$ and $1 \le j \le 2 \dsn(H)$.
	To this end, partition all direct-edges into
	$2\dsn(H)$ groups and employ $q$ stacks for each of the groups.
	A group of a direct-edge, $e$, with endpoints $(v, x)$ and $(u, y)$ is determined by the stack
	of $(v, u) \in E(H)$ in the dispersable layout of $H$ and by the relative order 
	of $x$ and $y$ in $\sigma$. Specifically, 
	\begin{compactitem}
		\item 	if $x <_{\sigma} y$, $(v, u) \in \Sh_j$, and $(x, y) \in \Qh_i$, then $e \in \Sh_i^{2j}$;
		\item 	if $y <_{\sigma} x$, $(v, u) \in \Sh_j$, and $(x, y) \in \Qh_i$, then $e \in \Sh_i^{2j+1}$.
	\end{compactitem}
	Here $\{\Sh_1, \dots, \Sh_{\dsn(H)}\}$ is the partition of $E(H)$ in the dispersable stack layout of $H$, 
	and $\{\Qh_1, \dots, \Qh_q\}$ is the partition of $E(G)$ in the $q$-queue layout of $G$.
	
	Let us verify that the direct-edges in a stack are non-crossing.	
	For the sake of contradiction, assume two edges, $e_1$ with endpoints
	$(v_1, x_1)$ and $(u_1, y_1)$, and $e_2$ with endpoints $(v_2, x_2)$ and $(u_2, y_2)$, cross each other.
	We assume $e_1$ and $e_2$ belong to a group $\Sh_i^{2j}$ for some $1 \le i \le q, 1 \le j \le \dsn(H)$;
	the other case is symmetric.	
	Since $e_1$ and $e_2$ cross, $\pi(v_1) = \pi(v_2)$ and
	$\phi(v_1, x_1) < \phi(v_2, x_2) < \phi(u_1, y_1) < \phi(u_2, y_2)$.
	By (C2), we have $x_1 <_{\sigma} x_2$, and by (C3), we have $y_2 <_{\sigma} y_1$.
	Hence, two edges of $G$ from the same queue, $(x_1, y_1)$ and $(x_2, y_2)$, form a $2$-rainbow in $\sigma$; 
	a contradiction.
	
	Therefore, $H \cbox G$ admits an $(s + \dsn(H))$-stack layout,
	$H \times G$ admits a $(2 q \cdot \dsn(H))$-stack layout, and $H \boxtimes G$ admits a 
	$(2 q \cdot \dsn(H) + s + \dsn(H))$-stack layout.
\end{proof}	

The bounds of \cref{thm:sn_product} can be improved for certain families of graphs.
For example, the stack number of the strong product of two paths is at most~$4$, 
while the theorem yields an upper bound of $7$; see \cref{fig:pathpath_stacks}. However, 
for a complete graph on $2k$ vertices, $K_{2k}$, it holds that 
$\sn(P_n \boxtimes K_{2k}) \ge 3k-1$ (following from the density of the product~\cite{BK79}), 
while $\sn(K_{2k}) = \qn(K_{2k}) = k$. Hence, the given bounds are asymptotically worst-case optimal.

Next we explore simultaneous linear layouts of bounded-pathwidth graphs.
While it is known that the stack number and the queue number of pathwidth-$p$ graphs is at most 
$p$~\cite{TY02,DMW05}, the existing proofs do not utilize the same vertex order for the stack
and queue layouts. We show that the bounds can be achieved in a simultaneous stack-queue layout.

\begin{lemma}
	\label{lm:pw_sim}
	A graph of pathwidth $p$ has a simultaneous $p$-stack $p$-queue layout.
\end{lemma}	

\begin{proof}
	Consider a vertex order, $\sigma$, of the given graph, $G$, corresponding to its
	vertex separation number, which equals to the pathwidth, $p$~\cite{Kin92,Bod98}.
	We prove that $\sigma$ yields a $p$-stack layout of $G$ and a $p$-queue layout of $G$;
	see \cref{fig:pw_sim}.
	
	Assume that edges of $G$ form a rainbow of size greater than $p$ with respect to $\sigma$.
	That is, let $\sigma$ be such that $u_1 <_{\sigma} \dots <_{\sigma} u_{p'} <_{\sigma} v_{p'} <_{\sigma} \dots <_{\sigma} v_1$
	for some $p' > p$ and $(u_i, v_i) \in E(G)$ for all $1 \le i \le p'$.
	Then the vertex cut at $v_{p'}$ has cardinality at least $p'$, 
	as $u_1, \dots, u_{p'} \in C(v_{p'})$,	
	which contradicts that the vertex separation
	is $p$. Therefore, the queue number of $G$ is at most $p$.
	
	To construct a $p$-stack layout, consider the vertices of $G$ in the order
	$v_1 <_{\sigma} v_2 <_{\sigma} \dots <_{\sigma} v_n$. Let $E^i$ be the set of
	\df{forward} edges of $v_i$, that is, 
	$E^i = \{(v_i, y) \in E(G) : v_i <_{\sigma} y\}$.
	We process the vertices in the order
	and assign edges to $p$ stacks while maintaining the following invariant 
	for every $1 < i \le n$:
	\begin{compactitem}
		\item all edges $E^1, \dots, E^{i-1}$ are assigned to one of $p$ stacks, and
		\item all edges from $E^j$ for every $1 \le j \le i-1$ are in the same stack.
	\end{compactitem}	
	Clearly, the invariant is satisfied for $i = 2$ by assigning $E^1$ to a single stack. 
	Suppose we obtained a stack assignment for all forward edges up to
	$E^{i-1}$; let us process $E^i$. Assume that $E^i \neq \emptyset$, and 
	observe that $v_i \in C(v_{i+1})$ and $|C(v_{i+1}) \setminus \{v_i\}| \le p - 1$. 
	Edges of $E^i$ can cross only already processed edges incident
	to a vertex from $C(v_{i+1}) \setminus \{v_i\}$. By the assumption of our invariant,
	such edges utilize at most $p-1$ distinct stacks. Hence, we have an available stack, which
	we use for $E^i$, and thus maintaining the invariant.
\end{proof}	

\begin{figure}[!tb]
	\begin{subfigure}[b]{0.49\textwidth}
		\centering
		\includegraphics[page=5,width=0.95\textwidth]{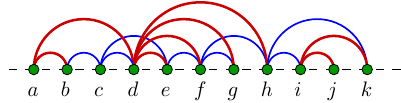}
		\caption{}
	\end{subfigure}
	\hfill
	\begin{subfigure}[b]{0.49\linewidth}
		\centering
		\includegraphics[page=6,width=0.95\textwidth]{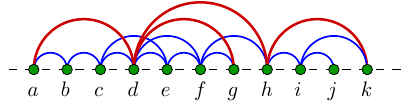}
		\caption{}
	\end{subfigure}
	\caption{An illustration for \cref{lm:pw_sim}: creating (a)~a $2$-stack and (b)~a $2$-queue layout 
		for a graph with pathwidth $2$ using its vertex separation order.}
\vspace{-0.3cm}
	\label{fig:pw_sim}
\end{figure}

\cref{cor:1} follows from \cref{thm:sn_product}, \cref{lm:pw_sim}, and an 
observation that $\dsn(P_n) = 2$. In order to prove \cref{cor:2}, we need the following
auxiliary lemma.

\begin{lemma}
	\label{lm:disp}
	Let $G$ be a bipartite graph of maximum vertex degree $\Delta$ that admits an
	$s$-stack layout. Then $\dsn(G) \le s \cdot \Delta$.
\end{lemma}	

\begin{proof}
	Edges of every stack of the $s$-stack layout of $G$ form an outerplanar graph. Since
	$G$ is bipartite, the edges of each stack can be partitioned into at most $\Delta$
	subgraphs, which are $1$-regular. Thus the dispersable stack number of $G$ is at most $s \cdot \Delta$.
\end{proof}	

\cref{cor:2} follows from \cref{thm:sn_product}, \cref{lm:disp}, and the fact that the stack number of
a graph with treewidth $tw$ is at most $tw + 1$~\cite{GH01}. Notice that in order to apply
\cref{thm:sn_product}, we set $G$ to be a given path and $H$ to be a given bipartite bounded-treewidth graph.
The bound of \cref{cor:2} can be reduced for low-treewidth graphs, whose dispersable stack number
is lower than the one given by \cref{lm:disp}~\cite{ABGKP18}.

\subsection{Negative Results}
In order to prove \cref{thm:sq_pw}, we first consider the case when $s=q=1$ and prove
the existence of a path decomposition of width $2$ with a certain property.

\begin{lemma}
	\label{lm:11_pw}
	Let $G$ be an $n$-vertex graph admitting a simultaneous
	$1$-stack $1$-queue layout with
	respect to a vertex order $\sigma = (v_1, v_2, \dots, v_n)$. 
	Then $G$ has pathwidth at most $2$. Furthermore, the corresponding path
	decomposition consists of $n$ bags $B_1, \dots, B_n$
	such that $|B_x| \le 3$ and $v_x \in B_x$ for all $1 \le x \le n$.
\end{lemma}

\begin{proof}
	It is tempting to approach the lemma by arguing that for a vertex order, $\sigma$, the
	corresponding vertex separation number is bounded. However, a simultaneous $1$-stack $1$-queue layout of a star graph
	with its center at position $\lfloor n/2 \rfloor$ of $\sigma$ has an unbounded vertex cut. Therefore, we
	explicitly construct a path decomposition of $G$ to prove the claim.
	We use induction on the number of vertices in $G$; the base of the induction with $n=1$ clearly holds.
	
	Consider the last vertex in the vertex order, $v_n \in V(G)$, and let $d \ge 0$ be the degree
	of $v_n$. If $d=0$ then we inductively construct a path decomposition for the first $n-1$ vertices
	and append a bag containing the single vertex, $\{v_n\}$. Thus we may assume $d > 0$.
	
	Let $v_i \in V(G)$ be the smallest (with respect to $\sigma$) neighbor of $v_n$ for some $1 \le i < n$.
	Since $\sigma$ corresponds to a simultaneous
	$1$-stack $1$-queue layout, no edges of $G$ cross each other and no edges of $G$ nest each other.
	Thus, every vertex $x \in V(G)$ with $v_i <_{\sigma} x <_{\sigma} v_n$ is either
	(a)~adjacent to~$v_i$, or 
	(b)~adjacent to~$v_n$, or
	(c)~adjacent to both $v_i$ and $v_n$, or
	(d)~an isolated vertex. Otherwise edge $(v_i, v_n)$ crosses or nests	
	an edge $(x, y)$ for some $y \in V(G) \setminus \{v_i, v_n\}$; see \cref{fig:sq}.
	
	In order to build a desired path decomposition, we inductively apply the argument to a subgraph
	of $G$ induced by the vertices $v_1, v_2, \dots, v_i$.
	Assume that the resulting path decomposition of the subgraph consists of bags $\{B_j : 1\le j \le i\}$.	
	We extend it to a path decomposition of $G$ by appending $n-i$ bags.	
	Namely, if $d > 1$ then the extension is
	$$\{v_i, v_{i+1}\}, \dots, \{v_i, v_{h-1}\}, \{v_i, v_h, v_n\}, \{v_{h+1}, v_n\}, \dots, \{v_{n-1}, v_n\}, \{v_n\},$$
	where $v_h, i < h < n$ is the first neighbor of $v_n$ after $v_i$ in $\sigma$.
	Otherwise if $d = 1$ then we use bags
	$$\{v_i, v_{i+1}\}, \dots, \{v_i, v_{n-1}\}, \{v_i, v_n\}.$$
	
	It is straightforward to verify that the constructed path decomposition of $G$ satisfies the 
	requirements of the lemma.
\end{proof}

\begin{figure}[!tb]
	\centering
	\includegraphics[page=1,width=0.5\textwidth]{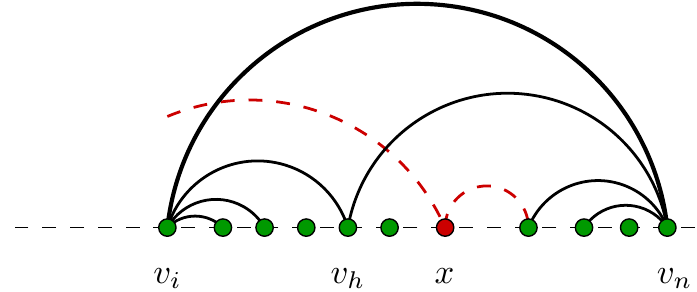}
	\caption{A graph admits a simultaneous $1$-stack $1$-queue layout if and only if
		its pathwidth is at most $2$, since a non-neighbor of $v_i$ and $v_n$, $x$, between the two vertices 
		creates either a crossing or a nested edge.}
	\label{fig:sq}
\end{figure}

\begin{backInTimeThm}{thm-pw}
	\begin{theorem}
		Let $G$ be a graph admitting a simultaneous $s$-stack $q$-queue layout. 
		Then $G$ has pathwidth at most $2 \cdot s \cdot q$.
	\end{theorem}	
\end{backInTimeThm}

\begin{proof}
	Consider a vertex order, $\sigma$, corresponding to 
	the simultaneous $s$-stack $q$-queue layout. Every edge of $G$ belongs to a stack and to a queue
	of the simultaneous layout. Thus all the edges can be partitioned into $s \cdot q$ disjoint sets,
	denoted $E^{i,j} \subseteq E(G)$ with $1 \le i \le s, 1 \le j \le q$, such
	that each set induces a simultaneous $1$-stack $1$-queue layout with vertex order $\sigma$.
	By \cref{lm:11_pw}, for every (possibly disconnected) subgraph $G^{i,j} = (V(G), E^{i,j})$ of $G$, there exists a 
	path decomposition of width~$2$ whose bags are denoted by $B_x^{i, j}$ for $x \in V(G)$.
	Define a path decomposition of $G$ to be $\{\cup_{i,j} B_x^{i,j} : x \in V(G)\}$, where
	the union is taken over all $1 \le i \le s, 1 \le j \le q$.	
	Next we verify that the construction is indeed a path decomposition of $G$:
	\begin{compactitem}
		\item Every edge of $G$ belongs to some set of the edge partition; thus, there is a bag 
		in the corresponding path decomposition, which contains both endpoints of the edge.
		
		\item For every $1 \le i \le s, 1 \le j \le q$, a vertex $x \in V(G)$ is in 
		a continuous interval of bags of the path decomposition of $G^{i,j}$.		
		By \cref{lm:11_pw}, the interval contains bag $B_x^{i,j}$; therefore, the union of 
		such intervals taken over all path decompositions forms a continuous interval.
	\end{compactitem}	
	
	Finally we notice that for all $1 \le i \le s, 1 \le j \le q$, bag $B_x^{i,j}, x \in V(G)$ 
	consists of vertex $x$ and possibly
	two more vertices of $G$. Hence, $|\cup_{i,j} B_x^{i,j}| \le 2 \cdot s \cdot q + 1$. 
	That is, the width of the constructed path decomposition
	is at most $2 \cdot s \cdot q$.
\end{proof}	

Observe that the bound of the above theorem may not be tight when $s > 1$ or $q > 1$. It is even
possible that for every graph $G$, $\pw(G)$ is linear in $(s + q)$.

Next our goal is to prove \cref{thm:separated}. To this end, we use an observation
by Erd\H{o}s and Szekeres~\cite{ES35} that 
for all $a, b \in \mathbb{N}$, every sequence of distinct numbers of length $a\cdot b + 1$ 
contains a monotonically increasing subsequence of length $a+1$ or 
a monotonically decreasing subsequence of length $b+1$. We start with an auxiliary lemma.

\begin{lemma}
	\label{lm:sep}
	Let $G$ be a graph with $2n$ vertices and $n$ independent edges $(u_i, v_i)$, $1 \le i \le n$.
	If $G$ admits an $s$-stack layout in which $u_1 < u_2 < \dots < u_n < v_i$ for all $1 \le i \le n$, then
	there exists a subgraph of $G$ with at least $r = \lceil n/s \rceil$ edges such that
	$u_{j_1} < u_{j_2} < \dots < u_{j_r} < v_{j_r} < \dots < v_{j_2} < v_{j_1}$ for some $1 \le j_1 < \dots < j_r \le n$.
\end{lemma}	

\begin{proof}
	Assume that the $s$-stack layout of $G$ is defined by the vertex order $\sigma$:
	$$u_1 <_{\sigma} u_2 <_{\sigma} \dots <_{\sigma} u_n <_{\sigma} v_{h_1} <_{\sigma} v_{h_2} <_{\sigma} \dots <_{\sigma} v_{h_n},$$
	where $h_1, h_2, \dots, h_n$ is a permutation of $\{1, 2, \dots, n\}$.
	An increasing subsequence of $h_1, h_2, \dots, h_n$ with length $k > 0$ corresponds to
	a $k$-twist in the stack layout of $G$. Hence, $k \le s$.	
	Now we can apply the result of Erd\H{o}s and Szekeres~\cite{ES35} for $a = s$ and $b = \lceil n/s \rceil - 1$,
	as $s (\lceil n/s \rceil - 1) + 1 \le n$ for all integers $n, s \ge 1$.
	Thus, the permutation contains a decreasing subsequence of length at least 
	$b + 1 = \lceil n/s \rceil$, which completes the proof.
\end{proof}	

Next we prove \cref{thm:separated}.
Recall that a stack layout of $P_n \boxtimes G$ is \df{separated}
if for two consecutive copies of $G$, denoted $G_1$ and $G_2$, all vertices of $G_1$ precede all vertices of $G_2$ in
the vertex order.

\begin{backInTimeThm}{thm-separated}
	\begin{theorem}
		Assume $P_n \boxtimes G$ has a separated layout on $s$ stacks. Then $G$ admits 
		a simultaneous $s$-stack $s^2$-queue layout, and therefore, $\pw(G) \le 2 s^3$.
	\end{theorem}	
\end{backInTimeThm}

\begin{figure}[!tb]
	\centering
	\includegraphics[page=2,width=0.75\textwidth]{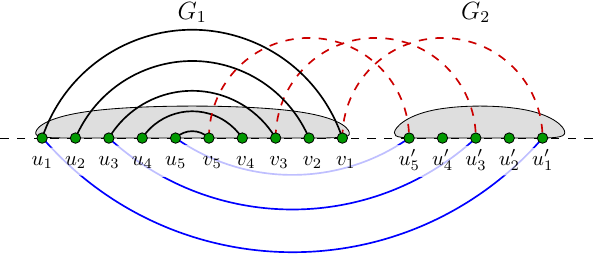}
	\caption{A large rainbow in a copy $G_1$ of a separated layout of 
		$P_n \boxtimes G$ yields a large twist formed by inter-copy edges between $G_1$ and $G_2$.}
	\label{fig:sep}
\end{figure}

\begin{proof}
	Suppose that $\sigma$ is a vertex order corresponding to the separated layout of $P_n \boxtimes G$, 
	and $G_1 = \big(V(G_1), E(G_1)\big), G_2 = \big(V(G_2), E(G_2)\big)$ are two copies of $G$ separated in $\sigma$.
	That is, $u <_{\sigma} v$ for all $u \in V(G_1), v \in V(G_2)$.
	
	Consider a suborder of $\sigma$ induced by the vertices of $V(G_1)$ and denote it by $\sigma_1$.
	If the largest rainbow formed by the edges of $E(G_1)$ with respect to $\sigma_1$ has the size at
	most $s^2$, then $\sigma_1$ corresponds to an $s^2$-queue layout~\cite{HR92}.
	In that case, $G$ admits a simultaneous $s$-stack $s^2$-queue layout using $\sigma_1$ as the vertex order, 
	which proves the claim of the theorem. Therefore, we may assume that the largest rainbow in $G_1$ is
	of size $k > s^2$; see \cref{fig:sep}.
	
	Let $(u_i, v_i) \in E(G_1), 1 \le i \le k$ be such a $k$-rainbow with 
	$$u_1 <_{\sigma} u_2 <_{\sigma} \dots <_{\sigma} u_k <_{\sigma} v_k <_{\sigma} \dots <_{\sigma} v_2 <_{\sigma} v_1.$$
	Consider vertices $u_1', \dots u_k' \in V(G_2)$ that are corresponding copies of $u_1, \dots u_k$ in graph $G_2$.
	Since the stack layout of $P_n \boxtimes G$ is separated, we have $u_1 <_{\sigma} \dots <_{\sigma} u_k <_{\sigma} u_i'$ for all
	$1 \le i \le k$. Hence, we may apply \cref{lm:sep} for a graph induced by vertices
	$u_1, \dots, u_k, u_1', \dots, u_k'$, which are connected by $k$ independent edges in the strong product.
	Therefore, we find a subset of $r = \lceil k/s \rceil > s$ edges in the graph such that
	$$u_{j_1} <_{\sigma} u_{j_2} <_{\sigma} \dots <_{\sigma} u_{j_r} <_{\sigma} u_{j_r}' <_{\sigma} \dots <_{\sigma} u_{j_2}' <_{\sigma} u_{j_1}',$$
	where $1 \le j_1, \dots, j_r \le k$.
	Finally, we observe that $(v_i, u_i'), 1 \le i \le k$ are direct-edges in the strong product, and vertices
	$$v_{j_r} <_{\sigma} v_{j_{r-1}} <_{\sigma} \dots <_{\sigma} v_{j_1} <_{\sigma} u_{j_r}' <_{\sigma} \dots <_{\sigma} u_{j_2}' <_{\sigma} u_{j_1}',$$
	form an $r$-twist in the $s$-stack layout of $G$. This contradicts to our assumption that 
	the largest rainbow in $G_1$ is of size greater than $s^2$.
	
	The bound on the pathwidth of $G$ follows from \cref{thm:sq_pw}.	
\end{proof}	

\section{Related Work}
\label{sect:related}

Although there exists numerous works on stack and queue layouts of graphs, 
layouts of graph products received much less attention. 
Wood~\cite{W05} considers queue layouts of various graph products, and
shows that the queue number of a product of graphs $H$ and $G$ is bounded by a function of the 
\df{strict} queue number of $H$ and the queue number of $G$. Here
a queue layout with an order $\sigma$ is \df{strict} if for no pair of edges, $(u, v)$ and $(x, y)$, it holds
that $u \le_{\sigma} x <_{\sigma} y \le_{\sigma} v$; see \cref{fig:strict}. 
Specifically, it is shown
that for all $H$ and $G$, $\qn(H \boxtimes G) \le 2 \sqn(H) \cdot \qn(G) + \sqn(H) + \qn(G)$,
where $\sqn(H)$ is the strict queue number of $H$. 
Similar bounds are given for the cartesian and direct products of $H$ and $G$.
It follows that $\qn(P_n \boxtimes G) \le 3 \qn(G) + 1$.
We stress that the result combined with a decomposition theorem for planar graphs~\cite{PS19,DJMMUW19} 
(such as one given by \cref{lm:kplanar_decom}) and the fact that the queue number of
planar 3-trees is bounded by a constant~\cite{ABGKP20}, yield a constant upper bound on the 
queue number of planar graphs.

\begin{figure}[!tb]
	\begin{subfigure}[b]{0.49\textwidth}
		\centering
		\includegraphics[page=3,width=0.9\textwidth]{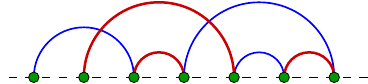}
		\caption{a dispersable $2$-stack layout}
		\label{fig:matching}
	\end{subfigure}
	\hfill
	\begin{subfigure}[b]{0.49\linewidth}
		\centering
		\includegraphics[page=4,width=0.9\textwidth]{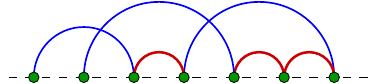}
		\caption{a strict $2$-queue layout}
		\label{fig:strict}
	\end{subfigure}
	\caption{Examples of a dispersable stack layout and a strict queue layout.}
\end{figure}

Stack layouts of graph products have also been studied~\cite{BK79,Over98,Che16,JSL18}, though most of the results 
are less complete as the problem is notoriously more difficult.
Bernhart and Kainen~\cite{BK79} introduce the concept of \df{dispersable} (also known as \df{matching}) 
book embeddings in which the graphs induced by the edges of each page are $1$-regular; see \cref{fig:matching}.
The minimum number of pages needed in a dispersable book embedding of $G$ is called its 
\df{dispersable stack number}, denoted $\dsn(G)$; it is also known as \df{matching book thickness}~\cite{ABGKP18}.
Clearly for every graph $G$ of maximum vertex degree $\Delta$, we have $\dsn(G) \ge \Delta$.
The authors of~\cite{BK79} observed that for every path, every tree,
every cycle of an even length, every complete bipartite graph, 
and every binary hypercube, it holds that $\dsn(G) = \Delta$. That made them conjecture that
the equation holds for every regular bipartite graph. The conjecture was disproved
in 2018 for every $\Delta \ge 3$ but it was shown that
$\dsn(G) = \Delta$ for every 3-connected 3-regular bipartite planar graph~\cite{ABGKP18}.

Bernhart and Kainen~\cite{BK79} show that for a bipartite graph $H$ and all 
graphs $G$, it holds that $\sn(H \cbox G) \le \dsn(H) + \sn(G)$; see~\cite{Over98}
for an alternative proof. Our \cref{thm:separated} generalizes the result.
Several subsequent papers study book embeddings of cartesian products 
for special classes of graphs~\cite{Li02,Che16,JSL18}; for example, when $H$ is a path and $G$ is a tree.
However to the best of our knowledge, no results on stack layouts of direct and strong products of graphs have
been published.

We remark that very recently, Dujmovi{\'c}, Morin, and Yelle~\cite{DMY20} independently proved
a result equivalent to \cref{cor:1}. Specifically, they study stack layouts of graphs with bounded
\emph{layered} pathwidth. That is, a path decomposition with a layering of a graph (a mapping 
$\ell: V(G) \rightarrow \mathbb{Z}$ such that $|\ell(u) - \ell(v)| \le 1$ for all $(u, v) \in E(G)$)
in which the size of the intersection of a bag and a layer is bounded by a constant.
It is shown that every graph of \emph{layered} pathwidth $p$ has stack number at most $4p$.
Since the strong product of a path and a pathwidth-$p$ graph has layered pathwidth $p+1$, the
result of~\cite{DMY20} implies (asymptotically) \cref{cor:1}. We emphasize that neither
our work nor \cite{DMY20} provides a tight bound on the stack number of the class of graphs.

\section{Conclusion}
\label{sect:concl}
In this paper we initiated the study of book embeddings of strong graph products. As explained in \cref{sect:intro}, 
resolving \cref{open:stack_strong_product} would either provide a constant upper bound on the stack number of
several families of non-planar graphs, or it would answer a fundamental question of Heath et al.~\cite{HLR92}
on the relationship of stack and queue layouts.

\cref{thm:separated} indicates that solving the open problem might be a challenging task. 
Thus we suggest to explore the problem for natural subclasses of bounded-treewidth graphs.

\begin{open}
	\label{open:stack_strong_product2}
	Is stack number of $P_n \boxtimes G$ bounded by a constant when $G$ is
	\begin{compactenum}[(i)]
		\item a tree (having an unbounded maximum degree)?
		\item an outerplanar (1-stack) graph?
		\item a planar graph with a constant treewidth, $\tw(G) \ge 2$?
		\item a bipartite graph with a constant treewidth, $\tw(G) \ge 2$?
	\end{compactenum}	
\end{open}

Notice that by the result of Wood~\cite{W05}, the queue number of $P_n \boxtimes G$ is a constant
for all the aforementioned graph families.

\bibliographystyle{abbrv}
\bibliography{refs_products}

\end{document}